\documentclass[10pt]{article}

\usepackage{graphics}
\usepackage{graphicx}
\usepackage{paralist}
\usepackage{color}
\usepackage{enumerate}
\usepackage{centernot}
\usepackage{listings}
\usepackage{color}
\usepackage{url}
\usepackage{bm}
\usepackage{bbm}
\usepackage{amssymb}
\usepackage{amsfonts}
\usepackage{amsmath}
\usepackage{fancyhdr}
\usepackage{fancyvrb}
\usepackage{natbib}
\usepackage[framed, thmmarks]{ntheorem}

\usepackage{afterpage}
\usepackage{paralist}
\usepackage{slashbox}
\usepackage{algorithm2e}
\usepackage[normalem]{ulem}

\newcommand\independent{\protect\mathpalette{\protect\independenT}{\perp}}
\def\independenT#1#2{\mathrel{\rlap{$#1#2$}\mkern2mu{#1#2}}}

\DeclareMathAlphabet{\mathsfit}{\encodingdefault}{\sfdefault}{m}{sl}

\bmdefine\aalpha{\alpha}
\bmdefine\ddelta{\delta}
\bmdefine\ttheta{\theta}
\bmdefine\mmu{\mu}
\bmdefine\pphi{\psi}
\bmdefine\ppsi{\psi}
\bmdefine\uupsilon{\upsilon}

\newtheorem{corollary}{Corollary}

\newtheorem{theorem}{Theorem}
\newtheorem{lemma}{Lemma}

\makeatletter
\newcommand{\openbox}{\leavevmode
  \hbox to.77778em{%
  \hfil\vrule
  \vbox to.675em{\hrule width.6em\vfil\hrule}%
  \vrule\hfil}}

\newcommand{\proofname}{Proof.}
\newcounter{proof}%
\newenvironment{proof}[1][\proofname]{
  \th@nonumberplain
  \def\theorem@headerfont{\itshape}%
  \normalfont
  \@thm{proof}{proof}{#1}}%
  {\@endtheorem}
\makeatother

\definecolor{lightgrey}{rgb}{0.9,0.9,0.9}
\definecolor{darkgreen}{rgb}{0,0.6,0}
\title{A parallel algorithm for ridge  penalized learning of the multivariate exponential family from data of mixed types}
\author{
{\small
\textbf{Diederik S. Laman Trip}$^{1}$, \small
\textbf{Wessel N. van Wieringen}$^{2,3,}$\footnote{Correspondence to: w.vanwieringen@amsterdamumc.nl}}
\\
{\small $^1$ Department of Bionanoscience, Kavli Institute of Nanoscience, Delft University of Technology}
\\
{\small Van der Maasweg 9, 2629 HZ Delft, The Netherlands}
\\
{\small $^2$ Department of Epidemiology and Data Science, Amsterdam UMC}
\\
{\small P.O. Box 7057, 1007 MB Amsterdam, The Netherlands}
\\
{\small $^3$ Department of Mathematics, Vrije Universiteit Amsterdam}
\\
{\small De Boelelaan 1111, 1081 HV Amsterdam, The Netherlands}
}

\begin{document}
\maketitle
% \date{}

\begin{abstract}
Computational efficient evaluation of penalized estimators of multivariate exponential family distributions is sought. These distributions encompass among others Markov random fields with variates of mixed type (e.g. binary and continuous) as special case of interest. The model parameter is estimated by maximization of the pseudo-likelihood augmented with a convex penalty. The estimator is shown to be consistent. With a world of multi-core computers in mind, a computationally efficient parallel Newton-Raphson algorithm is presented for numerical evaluation of the estimator alongside conditions for its convergence. Parallelization comprises the division of the parameter vector into subvectors that are estimated simultaneously and subsequently aggregated to form an estimate of the original parameter. This approach may also enable efficient numerical evaluation of other high-dimensional estimators. The performance of the proposed estimator and algorithm are evaluated and compared in a simulation study. Finally, the paper concludes with an illustration of the presented methodology in the reconstruction of the conditional independence network from data of an integrative omics study.
\\[3mm]
\end{abstract}

\section{Introduction}
With the increasing capacity for simultaneous measurement of an individual's many traits, networks have become an omnipresent visualization tool to display the cohesion among these traits. For instance, the cellular regulatory network portraits the interactions among molecules like mRNAs and/or proteins. Statistically, a network captures the relationships among variates implied by a joint probability distribution describing the simultaneous random behavior of the variates. These variates may be of different type, representing -- for example -- traits with continuous, count, or binary state spaces. Generally, the relationship network is unknown and is to be reconstructed from data. To this end we present methodology that learns the network from data with variates of mixed types in a computationally efficient manner.

A collection of $p$ variates of mixed type is mostly modeled by a pairwise Markov random field (MRF) distribution (a special case of the multivariate exponential family). A Markov random field is a set of random variables $Y_1, \ldots, Y_p$ that satisfies certain conditional independence properties specified by an undirected graph. This is made more precise by introduction of the relevant notions. A graph is a pair $\mathcal{G} = (\mathcal{V},\mathcal{E})$ with a finite set of vertices or nodes $\mathcal{V}$ and a collection of edges $\mathcal{E} \subseteq \mathcal{V} \times \mathcal{V}$ that join node pairs. In an undirected graph any edge is undirected, i.e. $(v_1, v_2) \in \mathcal{E}$ is an unordered pair implying that $(v_2, v_1) \in \mathcal{E}$. A subgraph $\mathcal{G}' \subseteq \mathcal{G}$ with $\mathcal{V}' \subseteq \mathcal{V}$ and $\mathcal{E}' \subseteq \mathcal{E}$ is a clique if $\mathcal{G}'$ is complete, i.e. all nodes are directly connected to all other nodes. The neighborhood of a node $v \in \mathcal{V}$, denoted $N(v)$, is the collection of nodes in $\mathcal{V}$ that are adjacent to $v$:  $N(v) = \{v' \in \mathcal{V} \, | \,  (v ,v') \in \mathcal{E}, v\not= v'  \}$. The closed neighborhood is simply $v \cup N(v)$ and denoted by $N[v]$. Now let $\mathbf{Y}$ be a $p$-dimensional random vector. Represent each variate of $\mathbf{Y}$ with a node in a graph $\mathcal{G}$ with $\mathcal{V} = \{1, \ldots. p\}$. Node names thus index the elements of $\mathbf{Y}$. Let $\mathcal{A}$, $\mathcal{B}$ and $\mathcal{C}$ be exhaustive and mutually exclusive subsets of $\mathcal{V} = \{1, \ldots. p\}$. Define the random vectors $\mathbf{Y}_a$, $\mathbf{Y}_b$ and $\mathbf{Y}_c$ by restricting the $p$-dimensional random vector $\mathbf{Y}$ to the elements of $\mathcal{A}$, $\mathcal{B}$ and $\mathcal{C}$, respectively. Then $\mathbf{Y}_a$ and $\mathbf{Y}_b$ are conditionally independent given random vector $\mathbf{Y}_c$, written as $\mathbf{Y}_a \independent \mathbf{Y}_b \, | \, \mathbf{Y}_c$, if and only if their joint probability distribution factorizes as $P(\mathbf{Y}_a, \mathbf{Y}_b \, | \, \mathbf{Y}_c) = P(\mathbf{Y}_a \, | \, \mathbf{Y}_c) \cdot P(\mathbf{Y}_b \, | \, \mathbf{Y}_c)$. The random vector $\mathbf{Y}$ satisfies the local Markov property with respect to a graph $\mathcal{G} = (\mathcal{V},\mathcal{E})$ if $Y_j \independent \mathbf{Y}_{\mathcal{V} \setminus N[j]} \, | \, \mathbf{Y}_{N(j)}$ for all $j \in \mathcal{V}$. Graphically, conditioning on the neighbors of $j$ detaches $j$ from $\mathcal{V} \setminus N[j]$. A Markov Random Field (or undirected graphical model) is a pair $(\mathcal{G}, \mathbf{Y})$ consisting of an undirected graph $\mathcal{G} = (\mathcal{V}, \mathcal{E})$ with associated random variables $\mathbf{Y} = \{Y_j \}_{j \in \mathcal{V}}$ that satisfy the local Markov property with respect to $\mathcal{G}$ (cf. \citealp{Lauritzen1996}). For strictly positive probability distributions of $\mathbf{Y}$ and by virtue of the Hammersley-Clifford theorem (\citealp{Hammersley1971}) the local Markov property may be assessed through the factorization of the distribution in terms of clique functions, i.e. functions of variates that correspond to a clique's nodes of the associated graph $\mathcal{G}$.

In this work we restrict ourselves to cliques of size at most two. Thus, only pairwise interactions between the variates of $\mathbf{Y}$ are considered. Although restrictive, many higher-order interactions can be approximated by pairwise interactions (confer, e.g., \citealp{Gallagher2011}). Under the restriction to pairwise interactions and the assumption of a strictly positive distribution, the probability distribution can be written as: 
\begin{eqnarray} \label{intro_mrf_simple}
P(\mathbf{Y}) & = & \exp \big[ - \mbox{$\sum_{j, j' \in \mathcal{V}}$} \, \phi_{j, j'}(Y_{j}, Y_{j'}) - D \big],
\end{eqnarray}
with log-normalizing constant or log-partition function $D$ and pairwise log-clique functions $\{\phi_{j, j'} \}_{j,j' \in \mathcal{V}}$. The pairwise MRF distribution $P(\mathbf{Y})$, and therefore the graphical structure, is fully known once the log-clique functions are specified. In particular, nodes $j, j' \in \mathcal{V}$ are connected by an edge whenever $\phi_{j, j'} \neq 0$ as the probability distribution of $\mathbf{Y}$ would then not factorize in terms of the variates, $Y_{j}$ and $Y_{j'}$, constituting this clique.

The estimation of the strictly positive MRF distribution (\ref{intro_mrf_simple}) with pairwise interactions will be studied here. This is hampered by the complexity of the log-partition function. Although analytically known for, e.g. the multivariate normal distribution, it is -- in general -- computationally not feasible to evaluate. Indeed, the partition function is computationally intractable for MRFs that have variables with a finite state space (\citealp{Welsh1993}, \citealp{Tibshirani2009}), or more generally for MRFs with variables of mixed type (\citealp{LeeHastie2013}). In effect, maximum likelihood estimation is prohibited computation. Instead parameters will estimation by means of pseudo-likelihood estimation. In particular, as the number of parameters is often of the same order -- if not larger -- as the sample size, the pseudo-likelihood will be augmented with a penalty. An overview of related work, which concentrates mainly on $\ell_1$-penalization, is supplied in the subsection at the end of the Introduction.

The contribution of this work to existing literature is three-fold. First, we present machinery for estimation of the mixed variate graphical model with a quadratic, i.e. ridge or $\ell_2$, penalty. Our motivation for ridge penalized estimation is multifold. For starters, it simply fills the lacuna in the literature which is mainly orientated towards lasso penalization. But it also deserves a place in the literature as, in contrast to lasso estimators, it yields -- due to its strict convex penalty -- a unique estimator. Related to the convexity issue, algorithms for the evaluation of lasso-type estimator often exhibit convergence problems, in particular when either the model is not extremely sparse, the data are medium to high-dimensional, and/or the penalization is not severe. Due to the smoothness and strict convexness of the ridge penalty, either an analytic expression of the estimator exists or a stable algorithm for the evaluation the corresponding estimator is conceivable. Moreover, biostatistical folkore has it that ridge estimators generally yield a better fit than those of lasso-type. This we have seen in the graphical model context, see e.g. \cite{van2016ridge,miok2017ridge,bilgrau2020targeted}. The simulations presented in those works also show that generally the performance of ridge estimation combined with posthoc selection compares well -- if not better -- to those of lasso estimation. Moreover, the smoothness and strict convexity of the ridge penalty can be used to allow other penalties (as was first pointed by \citealp{fan2001variable}). Within the graphical model context we have used this ourselves to derive a generalized lasso/elastic net estimator of the precision matrix of a Gaussian graphical model (see \citealp{van2019generalized}). Finally, the use of an ridge penalty translates the existing lasso-associated framework of sparse network estimation by allowing for estimation of dense networks. The need for such procedures can be found in the fact that the dominant paradigm of sparsity is not necessarily valid in all fields of application. In particular, in molecular biology this paradigm has recently come under fire (\citealp{Boyle2017}), and more dense (graphical) structures are advocated. Moreover, the severe degree of sparsity required by the theoretical resulting underpinning $\ell_1$ estimation procedures may lead to a too stringent inference on the underlying graph, potentially oversimplifying. Therefore, an $\ell_2$ estimation procedure might be more appropriate, especially in molecular biology applications, as it makes no sparsity assumption.

Another contribution is to be found in the efficient algorithm for the evaluation of the presented estimator. This exploits the high degree of parallelization allowed by modern computing systems. We developed a Newton-Raphson procedure that uses \textit{full} (instead of partial) second-order information with comparable computational complexity to existing methods that use only limited second-order information. Our approach translates to other high-dimensional estimators that may profit in their numerical evaluation.

Thirdly, the present work offers a software implementation to learn graphical models from data of more than two different variable types. A more pratical but no less relevant contribution as within the medical and biological field more and more different types of traits of samples are measured. This can be witnessed from the TCGA (The Cancer Genome Atlas) repository, where many types of molecular traits of cancer samples are measured. In a current development the aforementioned molecular information is augmented with imaging data (referred to as radiomics (\citealp{Gillies2015}). Not only more data, but also possibly of a different type. On top of this, efforts also aim to further augment these data with a quantification of a sample's `exposome', i.e. its environmental exposure (\citealp{Wild2012}). Hence, there is a need for methods and implementations that can deal with data comprising more than two types.

The paper is structured as follows. First, Section 2 the Markov random field distribution for variates of mixed types, along with parameter constraints that ensure its well-definedness, is recapped from literature as a special case of the more general exponential family. In Section 3 the exponential family model parameters, and thereby that of the MRF, are estimated by convex penalized pseudo-likelihood maximization, which is shown to yield a consistent estimator. The estimator is numerically evaluated by a form of the Newton-Raphson algorithm (Section 4). This algorithm is parallelized to exploit the multi-core capabilities of modern computing systems, and conditions that ensure convergence of the algorithm are identified. Section 5 presents \textit{a)} an \textit{in silico} comparison of the estimator to related ones and \textit{b)} a simulation study into the computational performance of the algorithm. This work concludes with a demonstration of the estimation procedure by re-analyzing data from an integrative omics study.

\subsection{Related work}
There is extensive literature on (learning) graphical models. Here we point to the sources relevant for the remainder. Attention originally focussed on a set of variates with a continuous state space following a multivariate normal distribution and thus giving rise to the Gaussian graphical model. The penalized estimation approaches of this model either maximize the likelihood augmented with a penalty (be it $\ell_1$ or $\ell_2$, (\citealp{Banerjee2008}, \citealp{Friedman2008}, \citealp{van2016ridge}), or comprise node-wise penalized regressions (\citealp{Meinshausen2006}, \citealp{Ravikumar2011}). The Ising model is the classic case for graphical models with discrete variables. Its estimation is prohibited by the computationally intractability of its partition function, which is generally the case for distributions of variates with a finite state space (e.g. Bernoulli). This may be circumvented by a stringent sparsity assumption (\citealp{Tibshirani2009}, \citealp{Lee2007efficient}). Alternatively, an approximation to the partition function gives rise to approaches that amount to node-wise regressions (\citealp{Ravikumar2010}, \citealp{wainwright2007high}, \citealp{Friedman2010}). This approach is extended to node-conditional distributions of multinomial or Poisson distributed variates (\citealp{jalali2011learning}, \citealp{allen2013local}). A final alternative would be to maximize the penalized pseudo-loglikelihood (\citealp{Tibshirani2009}, building on \citealp{besag1974spatial}), which has the advantage over the node-wise regression of producing a single parameter estimate for the whole graphical structure instead of multiple (possible contradictory) local ones. Moreover, reported simulations indicate that the resulting estimate is close to its penalized maximum likelihood counterpart (\citealp{Tibshirani2009}). Originating in \cite{besag1974spatial}, work on graphical models of mixed type focusses on the derivation the MRF distribution by the specification of the node-conditional distributions. \cite{Lauritzen1996} and  \cite{Lauritzen1989} produce the first examples for combinations of Bernoulli and Gaussian variates. Recently, \cite{yang2012graphical} extended this work and specified the node-conditional distribution of each variate as an univariate exponential family member, from which a corresponding joint MRF distribution was derived. These univariate exponential families include multi-parameter members with all but one parameter constant (e.g. Gaussian variables with constant variance). Subsequently, \cite{yang2012graphical} estimate the MRF parameter by a limited information approach exploiting the node-conditional distributions. Later \cite{yang2014mixed} generalized their work on MRF distributions to allow variates of mixed type with node-conditional distribution of any univariate exponential family member. An maximum $\ell_1$ penalized pseudo-likelihood estimator for this model, limited to a combination of Bernoulli and Gaussian variates, was put forward by \cite{LeeHastie2013}. In parallel effort, \cite{Chen2015} derived a similar joint MRF distribution that also includes multi-parameter exponential family members (e.g. Gaussian variables with unknown mean and variance), and specify conditions on its parameter space. However, after model specification, both \cite{yang2014mixed} and \cite{Chen2015} only consider variates of at most two different types in their estimation procedure.

\section{Model} \label{sect.model}
This section describes the graphical model for data of mixed types. In its most general form it is any exponential family distribution. Within the exponential family the model is first specified variate-wise, conditionally on all other variates. The parametric form of this conditionally formulated model warrants that the implied joint distribution of the variates is also an exponential family member. This correspondence between the variate-wise and joint model parameters endows the former (by way of zeros in the parameter) with a direct relation to conditional independencies between variate pairs, thus linking it to the underlying graph. Finally, parameter constraints are required to ensure that the proposed distribution is well-defined.

The multivariate exponential family is a broad class of probability distributions that describe the joint random behavior of a set of variates (possibly of mixed type). It encompasses many distributions for variates with a continuous, count and binary outcome space. All distributions share the following functional form:
\begin{eqnarray*}
f_{\mathbf{\Theta}}(\mathbf{y}) & = & h(\mathbf{y}) \exp \{ \eta(\mathbf{\Theta}) \hspace{1pt} T(\mathbf{y}) - D[\eta(\mathbf{\Theta})] \},
\end{eqnarray*}
where $\mathbf{\Theta}$ is a $p \times p$-dimensional parameter matrix, $h(\mathbf{y})$ is a non-negative base measure, $\eta(\mathbf{\Theta})$ is the natural or canonical parameter, $T(\mathbf{y})$ the sufficient statistic, and $D[\eta(\mathbf{\Theta})]$, the log-partition function or the normalization factor, which ensures $f_{\mathbf{\Theta}}(\mathbf{y})$ is indeed a probability distribution. The log-partition function $D[\eta(\mathbf{\Theta})]$ needs to be finite to ensure a well-defined distribution. For specific choices of $\eta, T$ and $h$, standard distributions are obtained. Theoretical results presented in Sections \ref{sect.estimation} and \ref{sect.algorithm} are stated for the multivariate exponential family, and thereby apply to all encompassing distributions. To provide for the envisioned practical purpose of reconstruction of the conditional dependence graph (as illustrated in Section \ref{sect.application}) we require and outline next a Markov random field in which the variates follow a particular exponential family member conditionally. This is thus a special case of the delineated class of exponential family distributions, as will be obvious from the parametric form of the Markov random field distribution.

Following \citealp{besag1974spatial} and \citealp{yang2014mixed} the probability distribution of each individual variate of $Y_{j}$ of $\mathbf{Y}$ conditioned on all remaining variates $\mathbf{Y}_{\setminus j}$ is assumed to be a (potentially distinct) univariate exponential family member, e.g. a Gaussian, exponential, or binomial distribution. Its (conditional) distribution is:
\begin{eqnarray} 
\label{form.nodeCondDist}
\quad P(Y_j \, | \, \mathbf{Y}_{\setminus j}) & \propto & h_j(Y_j) \exp \big[ \eta_{j}(\mathbf{\Theta}_{j,\setminus j}; \mathbf{Y}_{\setminus j}) \hspace{1pt} T_{j} (Y_j) - D_{j}( \eta_{j} ) \big].
\end{eqnarray}
Theorem \ref{thm.jointDist} below specifies the joint distribution for graphical models of variates that have a conditional distribution as in Display (\ref{form.nodeCondDist}). In particular, it states that there exists a joint distribution $P_{\mathbf{\Theta}}(\mathbf{Y})$ of $\mathbf{Y}$ such that $(\mathcal{G},\mathbf{Y})$ is a Markov random field if and only if each variate depends conditionally on the other variates through a linear combination of their univariate sufficient statistics. %It is analogous to Theorem (1) by \cite{yang2014mixed}.
\begin{theorem} (after \citealp{yang2014mixed}) \label{thm.jointDist} 
\\
Consider a $p$-variate random variable $\mathbf{Y} = \{Y_{j}\}_{j\in \mathcal{V}}$. Assume the distributions of each variate $Y_{j}$, $j \in \mathcal{V}$, conditionally on the remaining variates to be an exponential family member as in (\ref{form.nodeCondDist}). Let $\mathcal{G}=(\mathcal{V}, \mathcal{E})$ be a graph which decomposes into $\mathcal{C}$, the set of cliques of size at most two. Finally, the off-diagonal support of the MRF parameter $\mathbf{\Theta}$ matches the edge structure of $\mathcal{G}$.  Then, the following assumptions are equivalent:
\begin{compactitem}
\item[i)] For $j \in \mathcal{V}$, the natural parameter $\eta_{j}$ of the variate-wise conditional distribution (\ref{form.nodeCondDist}) is:
\begin{eqnarray}  \label{res_nat_param_full}
\eta_{j}(\mathbf{\Theta}_{j,\setminus j}; \mathbf{Y}_{\setminus j}) & = & \mathbf{\Theta}_{j, j} + \sum_{\{j' \in \mathcal{V} : (j, j') \in \mathcal{E}_{\mathsfit{C}}, \mathsfit{C} \in \mathcal{C} \}} \mathbf{\Theta}_{j, j'} \hspace{1pt} T_{j'} (Y_{j'}).
\end{eqnarray}
\item[ii)] There exists a joint distribution $P_{\mathbf{\Theta}}(\mathbf{Y})$ of $\mathbf{Y}$ such that $(\mathcal{G},\mathbf{Y})$ is a Markov random field.
\end{compactitem}
Moreover, by either assumption the joint distribution of $\mathbf{Y}$ is:
\begin{eqnarray} \label{res_joint_distr}
P_{\mathbf{\Theta}} (\mathbf{Y}) & \propto & \prod_{j \in  \mathcal{V}} h_{j} (Y_{j}) \exp \Big\{ T_{j}(Y_{j}) \Big[\mathbf{\Theta}_{j, j} + \sum_{\{j' \in \mathcal{V} : (j, j') \in \mathcal{E}_{\mathsfit{C}}, \mathsfit{C} \in \mathcal{C} \}}  \mathbf{\Theta}_{j, j'} \hspace{1pt} T_{j'} (Y_{j'}) \Big] \Big\}.
\end{eqnarray}
\end{theorem}
The theorem above differs from the original formulation in \citealp{yang2014mixed} in the sense that here it is restricted to pairwise interactions (i.e. cliques of size at most two).

For the reconstruction of the graph underlying the Markov random field, the edge set $\mathcal{E}$ is captured by the parameter $\mathbf{\Theta}$: nodes $j,j'  \in \mathcal{V}$ are connected by a direct edge $(j,j')\in \mathcal{E}$ if and only if $\mathbf{\Theta}_{j, j'} \neq 0$ (by the Hammersley-Clifford theorem, \citealp{Lauritzen1996}). This gives a simple parametric criterion to assess local Markov (in)dependence. Moreover, the parameter $\mathbf{\Theta}_{j,j'}$ can be interpreted as an interaction parameter between variables $Y_{j}$ and $Y_{j'}$.

We refer to the distribution from Display (\ref{res_joint_distr}) as the pairwise MRF distribution. After normalization of (\ref{res_joint_distr}), the joint distribution $P_{\mathbf{\Theta}}(\mathbf{Y})$ is fully specified by sufficient statistics and base measures of the exponential family members. For practical and illustrative purposes, the remainder will feature only four common exponential family members, the \textit{GLM family}: the Gaussian (with constant variance), exponential, Poisson and Bernoulli distributions. 

The joint distribution $P_{\mathbf{\Theta}}(\mathbf{Y})$ formed from the variate-wise conditional distributions need not be well-defined for arbitrary parameter choices. In order for $P_{\mathbf{\Theta}}(\mathbf{Y})$ to be well-defined, the log-normalizing constant $D[\eta(\mathbf{\Theta})]$ needs to be finite. For example, for the Gaussian graphical model, a special case of the pairwise MRF distribution under consideration, this is violated when the covariance matrix is singular. Lemma 1 of \citealp{Chen2015} specifies the constraints on the parameter $\mathbf{\Theta}$ that ensure a well-defined pairwise MRF distribution $P_{\mathbf{\Theta}}(\mathbf{Y})$ when the variates of $\mathbf{Y}$ are GLM family members conditionally. These constraints, i.e. for the GLM family, can be summarized in a table that is reproduced from \citealp{Chen2015} and -- for completeness -- included in Supplementary Material A. More general results for exponential family members may be derived, see e.g. \citealp{yang2014mixed}. 

The parameter constraints for a well-defined $P_{\mathbf{\Theta}}(\mathbf{Y})$ as identified by \cite{Chen2015} are restrictive on the structure of graph and the admissible interactions. As the graph is implicated by the off-diagonal support by $\mathbf{\Theta}$, the latter's constraints for well-definedness imply that the nodes corresponding to conditionally Gaussian random variables cannot be connected to the nodes representing exponential and/or Poisson random variables. Moreover, when $Y_{j}$ and $Y_{j'}$ are assumed to be Poisson and/or exponential random variables conditionally on the other variates, their interaction can only be negative. However, these restrictions can be relaxed by modeling data with, for example, a truncated Poisson distribution (\citealp{yang2014mixed}), or, while possibly not resulting in a well-defined joint distribution, might even be ignored altogether for network estimation (\citealp{Chen2015}).

Finally, the pairwise MRF distribution $P_{\mathbf{\Theta}}(\mathbf{Y})$ is well-defined on a convex parameter space (see Supplementary Material B). This implies that, when limiting ourselves to the class of pairwise MRF distributions for the GLM family and adhere to the constraints for well-definedness, parameter estimation amounts to a concave optimization problem with a convex domain.

\section{Estimation} \label{sect.estimation}
The parameter $\mathbf{\Theta}$ of the multivariate exponential family distribution $P_{\mathbf{\Theta}}(\mathbf{Y})$ is now to be learned from (high-dimensional) data. Straightforward maximization of the penalized loglikelihood is impossible due to the fact that the log-partition function cannot be evaluated in practice. For example, the partition function of the Ising model, which involves $p$ binary variates, sums over all $2^{p}$ configurations.  For large $p$ this becomes computationally intractable. Inspired by the work of \citealp{besag1974spatial} and \citealp{Tibshirani2009} this is circumvented by the replacement of the likelihood by the pseudo-likelihood comprising the variate-wise conditional distributions. We show that the maximum penalized pseudo-likelihood estimator of the exponential family model parameter is -- under conditions -- consistent. Finally, we present an algorithm for the numerical evaluation of this proposed estimator that efficiently addresses the computational complexity. Both results carry over to the pairwise MRF parameter as special case of the multivariate exponential family.

Consider an identically and independently distributed sample of $p$-variate random variables $\mathbf{Y}_1, \ldots, \mathbf{Y}_n$ all drawn from $P_{\mathbf{\Theta}}$. The associated (sample) pseudo-loglikelihood is a composite loglikelihood of all variate-wise conditional distributions averaged over the observations:
\begin{eqnarray} 
\label{res_def_sample_PL}
\mathcal{L}_{PL}(\mathbf{\Theta}, \mathbf{Y}_1, \ldots, \mathbf{Y}_n) &  = & \frac{1}{n}\sum_{i=1}^n \sum_{j \in \mathcal{V}} \log [ P_{\mathbf{\Theta}} (Y_{ij} \, | \, \mathbf{Y}_{i, \setminus j})].
\end{eqnarray}
The maximum penalized pseudo-loglikelihood augments this by a strictly convex, continuous penalty function $f^{\mbox{{\tiny pen}}}(\mathbf{\Theta}; \lambda)$ with penalty parameter $\lambda > 0$, hence \newline $\mathcal{L}_{\mbox{{\tiny penPL}}}(\mathbf{\Theta}, \mathbf{Y}_1, \ldots, \mathbf{Y}_n)  := \mathcal{L}_{PL}(\mathbf{\Theta}, \mathbf{Y}_1, \ldots, \mathbf{Y}_n) - f^{\mbox{{\tiny pen}}}(\mathbf{\Theta}; \lambda)$. Then the maximum penalized pseudo-likelihood estimator of $\mathbf{\Theta}$ is:
\begin{eqnarray} \label{form.ridgePseudoEstimator}
\widehat{\mathbf{\Theta}}^{\mbox{\tiny pen}}(\lambda)  & = & \arg \max_{\mathbf{\Theta}}  \mathcal{L}_{\mbox{{\tiny penPL}}} (\mathbf{\Theta}, \mathbf{Y}_1, \ldots, \mathbf{Y}_n).
\end{eqnarray}
When $f^{\mbox{{\tiny pen}}}(\mathbf{\Theta}; \lambda)$ is proportional to the sum of the square of the elements of the parameter, $f^{\mbox{{\tiny pen}}}(\mathbf{\Theta}; \lambda) = \tfrac{1}{2} \lambda \| \mathbf{\Theta} \|_F^2$ with $\| \cdot \|_F$ the Frobenius norm, it is referred to as the ridge penalty. With the ridge penalty, the estimator (\ref{form.ridgePseudoEstimator}) is called the maximum ridge pseudo-likelihood estimator. We show that the maximum penalized pseudo-likelihood estimator (\ref{form.ridgePseudoEstimator}) is consistent in the traditional sense, i.e. a regime of fixed dimension $p$ and an increasing sample size $n$.  Note that, in practice -- as recommended by \cite{Tibshirani2009} -- we employ $f^{\mbox{{\tiny pen}}}(\mathbf{\Theta}; \lambda) = \tfrac{1}{2} \lambda \sum_{j, j'=1, j \not= j'}^p [ (\mathbf{\Theta})_{j,j'}]^2$, thus leaving the diagonal unpenalized.  Empirically, we observed this yields a better model fit, which is intuitively understandable as the estimator is then able to (unconstrainedly) account for (at least) the marginal variation in each variate. The theoretical results below are, however, presented for $f^{\mbox{{\tiny pen}}}(\mathbf{\Theta}; \lambda) = \tfrac{1}{2} \lambda \| \mathbf{\Theta} \|_F^2$, as this simplifies the mathematical proof, which can be modified for the case with an unpenalized diagonal.

The motivation for a traditional consistency result, i.e. in the $p$ fixed while $n \rightarrow \infty$ regime, is three-fold. 
\begin{compactitem}
\item[\textit{1)}] It suffices for practical purposes in the `omics-field' where the system's dimensions are known and fixed at the outset of the analysis. 

\item[\textit{2)}] Almost all existing high-dimensional results (cf. e.g. \citealp{LeeHastie2013}, \citealp{Chen2015}, \citealp{Lee2015}) show `sparsistency' -- a contraction of sparse and consistency -- which amounts to the probability of correct (sparse) network structure selection tending to one as $p \rightarrow \infty$. Sparsistency thus addresses the selection property of an estimator but not the (limiting) quality of the estimator in terms of the parameter values. In this light sparsistency results have -- understandably -- only been proven for $\ell_1$-penalized (or a generalization thereof) estimators (e.g., \citealp{Lee2015}). Under the assumptions of \textit{i)} strong convexity of the loss function  and \textit{ii)} irrepresentability of the model (for details see \citealp{Lee2015}), sparsistency is proven for sparse models (with the maximum vertex degree $d$ such that $d=o(n)$, \citealp{Chen2015}), with a minimal parameter value (in the absolute sense) for those parameters being nonzero, and a minimum sample size dependent on the dimension $p$ and the degree of sparsity (\citealp{Lee2015}). These assumptions and requirements are hard -- if not impossible -- to verify for any practical case at hand  (\citealp{Bento2009}, \citealp{Anandkumar2012}), thus questioning the usefulness of sparsistency results (\citealp{LeeHastie2013}). Moreover, for the proposed ridge estimator, which does not select, sparsistency appears not to be the relevant concept. 

\item[\textit{3)}] An information theoretic analysis shows that learning graphical models, even for simple cases such as the Ising or Gaussian graphical model, requires at least $n=\Omega(d^2\log p)$ samples (in words, $n$ is at least of order $d^2 \log p$) as $p\to\infty$ (\citealp{das2012learning}). For example, when $p > n$, one finds that $p > n > d^2\log p$. This, when solved for $d$, yields: $d < \sqrt{p / \log p}$. For $p=100$ the maximum vertex degree $d$ would then be bounded by $4.6$, which rules out dense networks. Therefore, there is no hope for successful graphical model selection in the high-dimensional setting ($n<p$) when the maximum vertex degree satisfies $d = \Omega(\sqrt{p/\log p})$, which can be assumed to be the case for molecular biology applications or scale-free and dense networks in general.
\end{compactitem}
Hence, the focus on traditional consistency which is a minimum requirement of a novel estimator.

\begin{theorem} \label{thm.consistency}
Let  $\mathbf{Y}_1, \ldots, \mathbf{Y}_n$ be $p$-variate draws from a exponential family distribution $P_{\mathbf{\Theta}}(\mathbf{Y}) \propto \exp[\mathbf{\Theta} \, T(\mathbf{Y}) + h(\mathbf{Y})]$. Temporarily supply $\widehat{\mathbf{\Theta}}$ and $\lambda$ with an index $n$ to explicate their sample size dependence. Then the maximum  penalized pseudo-likelihood estimator $\widehat{\mathbf{\Theta}}_n^{\mbox{{\tiny pen}}}$ maximizing the penalized pseudo-likelihood is consistent, i.e., $\widehat{\mathbf{\Theta}}_n^{\mbox{{\tiny pen}}} \stackrel{p}{\longrightarrow}\mathbf{\Theta}$ as $n \rightarrow \infty$ if,

\begin{compactitem}
\item[i)] The parameter space is compact and such that $P_{\mathbf{\Theta}}(\mathbf{Y})$ is well-defined for all $\mathbf{\Theta}$,
\item[ii)] $\mathbf{\Theta} \, T(\mathbf{Y}) + h(\mathbf{Y})$ can be bounded by a polynomial, $|\mathbf{\Theta} \, T(\mathbf{Y}) + h(\mathbf{Y})| \leq c_1 + c_2 \sum_{j \in \mathcal{V}} |Y_{j}|^\beta$ for constants $c_1,c_2<\infty$ and $\beta\in\mathbb{N}$,
\item[iii)] The penalty function $f^{\mbox{{\tiny pen}}}(\mathbf{\Theta})$ is strict convex, continuous, and the penalty parameter $\lambda_n$ converges in probability to zero: $\lambda_n \stackrel{p}{\longrightarrow} 0$ as $n \rightarrow \infty$.
\end{compactitem}
\end{theorem}

\begin{proof}
Refer to Supplementary Material B.
\end{proof}

Theorem \ref{thm.consistency} differs from related theorems on $\ell_1$-estimators in two respects. Most importantly, \textit{i)} it holds uniformly over all (well-defined) models, i.e. it does not require a sparsity assumption. Morever, \textit{ii)} the assumption on the penalty parameter is of a probabilistic rather than a specific deterministic nature, which we consider to be more suited as $\lambda$ is later chosen in a data-driven fashion.

Theorem \ref{thm.consistency} warrants -- under conditions -- the convergence of the maximum  penalized pseudo-likelihood estimator $\widehat{\mathbf{\Theta}}$ as the sample size increases $n\to\infty$. These conditions require a compact parameter space, a common assumption in the field of graphical models (\citealp{Lee2015}). Theorem \ref{thm.consistency} holds in general for any multivariate exponential family distribution and is therefore generally applicable with the pairwise MRF distribution as special case. Hence, when $P_{\mathbf{\Theta}}(\mathbf{Y})$ is well-defined for the GLM family, we obtain the following corollary.

\begin{corollary} \mbox{ }
Let $\mathbf{Y}_1, \ldots,\mathbf{Y}_n$ be $p$-variate draws from a well-defined pairwise MRF distribution  $P_{\mathbf{\Theta}}(\mathbf{Y})$ with parameter $\mathbf{\Theta}$.  The ridge pseudo-likelihood estimator $\widehat{\mathbf{\Theta}}_n^{\mbox{{\tiny ridge}}}$ that maximizes the ridge-penalized pseudo-likelihood is consistent, i.e.,
$\widehat{\mathbf{\Theta}}_n^{\mbox{{\tiny ridge}}} \stackrel{p}{\longrightarrow}\mathbf{\Theta}$ as $n \rightarrow \infty$, if the parameter space is compact, and the penalty parameter $\lambda_n$ converges in probability to zero:  $\lambda_n \stackrel{p}{\longrightarrow} 0$ as $n \rightarrow \infty$.
\end{corollary}

\begin{proof}
Refer to Supplementary Material B.
\end{proof}

\section{Algorithm} \label{sect.algorithm}
Maximization of the ridge pseudo-loglikelihood presents due to a convex optimization problem (a concave pseudo-loglikelihood and convex parameter space). To this end we present a parallel block coordinate Newton Raphson algorithm for numerical evaluation of the penalized pseudo-loglikelihood estimator $\widehat{\mathbf{\Theta}}(\lambda)$. We show that this algorithm yields a sequence of updated parameters that converge to $\widehat{\mathbf{\Theta}}(\lambda)$ and terminates after a finite number of iterations. The results in this section do not require the parameter constraints for well-definedness of Section \ref{sect.model} and thus hold for any multivariate exponential family.

The strict concavity of the optimization problem (\ref{form.ridgePseudoEstimator}) and the smoothness of $\mathcal{L}_{\mbox{{\tiny penPL}}}$ permit the application of the Newton-Raphson algorithm to find the estimate. The Newton-Raphson algorithm starts with an initial guess $\widehat{\mathbf{\Theta}}^{(0)}(\lambda)$ and -- motivated by a Taylor series approximation -- updates this sequentially. This generates a sequence $\{ \widehat{\mathbf{\Theta}}^{(k)}(\lambda) \}_{k \geq 0}$ that converges to $\widehat{\mathbf{\Theta}} (\lambda)$ (\citealp{fletcher2013}). However, the Newton-Raphson algorithm requires inversion of the Hessian matrix and is reported to be slow for pseudo-loglikelihood maximization (\citealp{LeeHastie2013}, \citealp{Chen2015}): it has computational complexity $O(p^6)$ for $p$ variates. Instead of a naive implementation of the Newton-Raphson algorithm to solve (\ref{form.ridgePseudoEstimator}), the remainder of this section describes a block coordinate approach (inspired by \citealp{Xu2013}) that speeds up the evaluation of the estimator by exploiting the structure of the pseudo-likelihood and splitting the optimization problem (\ref{form.ridgePseudoEstimator}) into multiple simpler subproblems. The novelty of this parallel block coordinate Newton-Raphson algorithm is necessary to answer to the ever increasing size of data sets and make optimal use of available multi-core processing systems. Finally, in contrast to other pseudo-likelihood approaches by \citealp{Tibshirani2009} or \citealp{LeeHastie2013}, the presented approach allows for the use of all second-order information (i.e. the Hessian) without much increase in computational complexity and the benefit of potentially faster convergence.

In order to describe the block coordinate approach some notation is introduced. Define $q=\frac{1}{2}p(p+1)$, the number of unique parameters of $\mathbf{\Theta}$. The set of unique parameter indices is denoted by $\mathcal{Q} =  \{(j, j') \, : \,  j \leq j' \in \mathcal{V} \}$ and we use $\ttheta$ as shorthand for the $q$-dimensional vector of unique parameters $\{\mathbf{\Theta}_{j, j'} \}_{(j, j') \in \mathcal{Q} }$. Furthermore, write $\ttheta_j$ for $\mathbf{\Theta}_{\ast,j} = (\mathbf{\Theta}_{j,\ast})^{\top}$, the $p$-dimensional vector of all unique parameters of $\mathbf{\Theta}$ that correspond to the $j$-th variate. Consequently, for $j \not= j'$ the corresponding $\ttheta_j$ and $\ttheta_{j'}$ have parameter(s) of $\mathbf{\Theta}$ in common. Finally, let $\mathbf{H}_{j}$ be the $p\times p$-dimensional submatrix of the Hessian limited to the elements that relate to the $j$-th variate, i.e.: $\mathbf{H}_{j}  =   \partial^2 \mathcal{L}_{\mbox{{\tiny penPL}}} / \partial \ttheta_j \partial \ttheta_j^{\top}$.

In our block coordinate approach we maximize the penalized pseudo-loglikelihood with respect to the parameter subvector $\ttheta_j$ for $j \in \mathcal{V}$ while all other parameters are temporarily kept constant at their current value. Per block we maximize by means of the Newton-Raphson algorithm starting from the initial guess $\hat{\ttheta}^{(0)}(\lambda)$ and updating the current parameter value $\hat{\ttheta}_j^{(k)}(\lambda)$ by $\hat{\ttheta}_j^{(k+1)}(\lambda)$ through:
\begin{eqnarray} 
\label{res_linear_feature_systems}
\hat{\ttheta}_j^{(k+1)}(\lambda) & = &  \hat{\ttheta}_j^{(k)}(\lambda) - \left. \Big( \frac{\partial^2 \mathcal{L}_{\mbox{{\tiny penPL}}} }{ \partial \ttheta_j  \partial \ttheta_j^{\top} } \Big)^{-1} \right|_{\ttheta = \hat{\ttheta}^{(k)}(\lambda)} \times \left. \frac{\partial \mathcal{L}_{\mbox{{\tiny penPL}}} }{ \partial \ttheta_j } \right|_{\ttheta = \hat{\ttheta}^{(k)}(\lambda)}. 
\end{eqnarray}
Block-wise the procedure converges to the optimum, that is, the maximum of $\mathcal{L}_{\mbox{{\tiny penPL}}}$ given the other parameters of $\ttheta$. Sequential application of the block coordinate approach is -- by the concavity of $\mathcal{L}_{\mbox{{\tiny penPL}}}$ -- then guaranteed to converge to the desired estimate. Sequential application of the block coordinate approach may be slow and is ran here in parallel for all $j \in \mathcal{V}$ simultaneously. This yields $\{\hat{\ttheta}_j^{(k+1)} \}_{j \in \mathcal{V} }$. But as some elements of $\ttheta_{j}$ and $\ttheta_{j'}$ map to the same element of $\ttheta$, multiple estimates of the latter are thus available. Hence, the results of each parallel iteration need to be combined in order to provide a single update of the full estimate $\hat{\mathbf{\ttheta}}^{(k)}$. This update of $\hat{\ttheta}^{(k)}$ should increase $\mathcal{L}_{\mbox{{\tiny penPL}}}$ and iteratively solve the concave optimization problem (\ref{form.ridgePseudoEstimator}). We find such an update in the direction of the sum of the block-wise updates of $\{\hat{\ttheta}_j^{(k+1)} \}_{j\in\mathcal{V}}$. A well-chosen step size in this direction then provides a suitable update of $\hat{\ttheta}^{(k)}$. Alternatively, to avoid the need for combining block-wise updates one may seek a split of the elements of $\mathbf{\Theta}$ into blocks without overlap. This, however, raises two issues. First, at each iteration $\mathbf{\Theta}$ is effectively estimated by a series of parallelly executed regressions. These regression share parameter. One may then use the estimate on a shared parameter from one equation in the other. But as the algorithm is parallelized, one can only use the estimate from the previous iteration, which will affect the convergence. Secondly, it requires a choice: which equation provides the estimate for a shared parameter. There is no obvious rationale that tells which one should prevail.

% \newpage 
\begin{figure*}[h!]
\begin{center}
\centering
\begin{tabular}{c}
{\includegraphics[angle=0,scale=0.8,clip,trim=40 510 40 50 ]{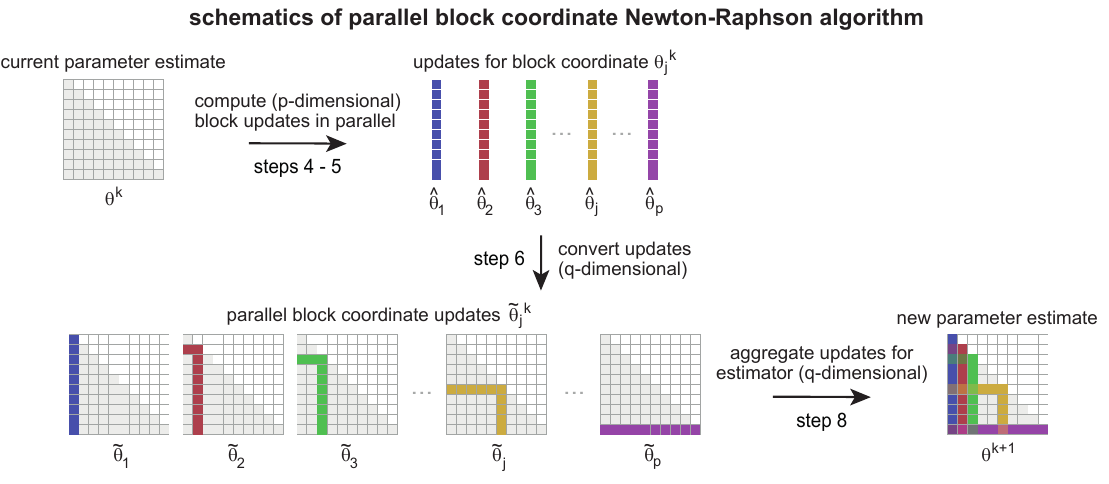}}
\end{tabular}
\end{center}
\vspace{-0.5cm}
\caption{Parameter updating by Algorithm \ref{alg.pseudocode}. First, the $p$ subvectors $\ttheta_j$ are updated to novel $\hat{\ttheta}_j$ (step 5 of Algorithm \ref{alg.pseudocode}). These updates are interspersed with zero's to form the $q$-dimensional vectors $\tilde{\ttheta}_j$ (step 6 of Algorithm \ref{alg.pseudocode}). Finally, the $\tilde{\ttheta}_j$'s are averaged weightedly to produce the update of $\ttheta$ (step 8 of Algorithm \ref{alg.pseudocode}).
}
\label{fig.parameterUpdate}
\end{figure*}
\afterpage{\clearpage}

Algorithm \ref{alg.pseudocode} gives a pseudo-code description of the parallel block coordinate Newton-Raphson algorithm (see Figure \ref{fig.parameterUpdate} for a visualization of the combination of the block-wise estimates). Theorem \ref{thm.convergenceAlgorithm} states that Algorithm \ref{alg.pseudocode} converges to the maximum  penalized pseudo-likelihood estimator and terminates. While Theorem \ref{thm.convergenceAlgorithm} is a rather general result for the maximum penalized pseudo-likelihood estimator of exponential family distributions, as special case the same result follows for the pairwise MRF distribution with the GLM family.

\begin{theorem} \label{thm.convergenceAlgorithm}
Let $\mathbf{Y}_1,\ldots,\mathbf{Y}_n$ be $n$ independent draws from a $p$-variate exponential family distribution $P_{\mathbf{\Theta}} \, (\mathbf{Y})\propto\exp[\mathbf{\Theta} \, T(\mathbf{Y}) + h(\mathbf{Y})]$. Assume that the parameter space of $\mathbf{\Theta}$ is compact. Let $\widehat{\mathbf{\Theta}}(\lambda)$ be the unique global maximum of the penalized pseudo-likelihood $\mathcal{L}_{\mbox{\tiny penPL}}(\mathbf{\Theta},\mathbf{Y}_1,\ldots,\mathbf{Y}_n)$. Then, for any initial parameter $\ttheta^{(0)}$, threshold $\tau>0$ and sufficiently large multiplier $\alpha\geq p$, Algorithm \ref{alg.pseudocode} terminates after a finite number of iterations and generates a sequence of parameters $\{\ttheta^{(k)}\}_{k\geq 0}$ that converge to $\widehat{\mathbf{\Theta}}(\lambda)$.
\end{theorem}

\begin{proof}
Refer to Supplementary Material C.
\end{proof}

% \vfill\null 
% \columnbreak
The presented Algorithm \ref{alg.pseudocode} balances computational complexity, convergence rate and optimal use of available information. The algorithm terminates after a finite number of iterations and one iteration, i.e. lines 3 to 9, has computational complexity $O(p^3)$ when run in parallel. Moreover, Algorithm \ref{alg.pseudocode} uses all available second-order information (the Hessian of $\mathcal{L}_{\mbox{\tiny penPL}}$) and its convergence rate is at least linear. However, the convergence rate is quadratic when the multiple updates for each parameter are identical.

% \newpage
\begin{minipage}[h]{\textwidth}
\begin{center}
\fbox{\parbox{14cm}{
% \removelatexerror
\begin{algorithm}[H]
\SetKwInOut{Input}{input}
\SetKwInOut{Output}{output}
\RestyleAlgo{boxed}
\LinesNumbered
\SetKwFor{inpar}{for}{do in parallel}{\text{ end synchronize}}
\Input{$n\times p$ data matrix $\mathbf{Y}$;
\\
$p$ exponential family members;
\\
initial parameter $\ttheta^{(0)}$;
\\
penalty parameter $\lambda_n \in \mathbb{R}_{>0}$;
\\
threshold $\tau \in \mathbb{R}_{> 0}$;
\\
step size $\alpha > 0$.
}
\Output{sequence $\{\ttheta^{(k)}\}_{k \geq 0}$.}
\vspace{0.25cm}
{\bf initialize} {$k=0$, ${err}_0 =2\tau$}.\\
\vspace{0.25cm}
\While{ ${err}_k > \tau$ }{
\BlankLine
\inpar{$j \in \mathcal{V}$}{
\BlankLine
calculate the gradient $\partial \mathcal{L}_{\mbox{{\tiny penPL}}} / \partial \ttheta_{j}$ and Hessian $\mathbf{H}_{j}$.
\\
\BlankLine
compute a single Newton-Raphson update of $\ttheta_j$.
\BlankLine
formulate the update as a $q$-dimensional vector $\tilde{\ttheta}_{j}$ by:
\begin{eqnarray*}
(\tilde{\ttheta}_j)_{\mathsfit{q}} & = &
\left\{
\begin{array}{cl}
(\ttheta_j)_{j'} & \mbox{for } \mathsfit{q} \in \mathcal{Q} \mbox{ s.t. } \mathsfit{q} = (j, j') \mbox{ or } \mathsfit{q} = (j', j),
\\
0 & \mbox{otherwise.}
\end{array}
\right.
\end{eqnarray*}
% \BlankLine
}
\BlankLine
define the parameter estimate $\hat{\ttheta}^{(k+1)} := \hat{\ttheta}^{(k)} + \frac{1}{\alpha}\sum_{j \in \mathcal{V}} \tilde{\ttheta}_j$.
\\
\BlankLine
assess error
${err}_k = \| \left. \partial \mathcal{L}_{\mbox{{\tiny penPL}}} / \partial \ttheta \right|_{\ttheta = \hat{\ttheta}^{k+1}} \|_2$ and $k = k+1$.\BlankLine
}\BlankLine
\vspace{0.25cm}
\caption{{\small Pseudocode of the parallel block coordinate Newton-Raphson algorithm for evaluation of the penalized pseudo-likelihood estimator.}} \label{alg.pseudocode}
\end{algorithm}
}
}
\end{center}
\end{minipage}
% \clearpage

\mbox{ } 
\\
\\
The pseudo-likelihood method has previously been reported to be computationally intensive with slow algorithms (\citealp{Chen2015}). For instance, the computational complexity of pseudo-likelihood maximization is $O(p^6)$ per iteration for a naive implementation of the Newton-Raphson algorithm. Comparable work therefore uses either the pseudo-likelihood or a node-wise regression. When maximizing the pseudo-loglikelihood, existing methods use a diagonal Hessian or an approximation thereof, or only first-order information (\citealp{Tibshirani2009}, \citealp{LeeHastie2013}). Such approaches achieve linear convergence at best and have a computational complexity of at least $O(np^2)$ per iteration as the gradient of the pseudo-loglikelihood must be evaluated. Alternatively, the computational complexity of node-wise regression methods is $O(p^4)$ per iteration for existing algorithms, which could be optimized to $O(p^3)$ with a parallel implementation. However, node-wise regression methods estimate each parameter twice and subsequently need to aggregate their node-wise estimates. This aggregated estimate does not exhibit quadratic convergence. Moreover, these node-wise estimates are potentially contradictory and their quality depends on the type of the variable (\citealp{Chen2015}). In summary, we expect Algorithm (\ref{alg.pseudocode}) to perform no worse than other pseudo-likelihood maximization approaches, since its computational complexity of $O(p^3)$ is comparable or better than existing methods and all available second-order information is used.

Finally, the condition on the multiplier $\alpha$ in Theorem \ref{thm.convergenceAlgorithm} may be relaxed (cf. Lemma \ref{supp_thm_non_conv_comb}), thereby appropriately increasing the step size of the parameter update and the convergence speed of Algorithm \ref{alg.pseudocode}.  

\begin{lemma} \label{supp_thm_non_conv_comb} \mbox{ } 
Let $\mathbf{Y}=\{\mathbf{Y}_1,...,\mathbf{Y}_n\}$ be $p$-variate draws from an exponential family distribution \newline $P_{\mathbf{\Theta}}(\mathbf{Y})\propto \exp[\mathbf{\Theta} \, T(\mathbf{Y}) + h(\mathbf{Y})]$. Let $\ttheta^{(0)}$ be any initial parameter unequal to the maximum ridge-penalized pseudo-likelihood estimator $\widehat{\mathbf{\Theta}}^{\mbox{{\tiny ridge}}}(\lambda)$. A single iteration of Algorithm (1) initiated with $\ttheta^{(0)}$ yields the block solutions $\{\ttheta_j \}_{j \in \mathcal{V} }$ (Line 5, Algorithm 1) and block-wise updates $\{ \tilde{\ttheta}_j \}_{j \in\mathcal{V}}$ (Line 6, Algorithm 1). Let $\alpha > 0$ and define $\ttheta^{(1)}$ as $\ttheta^{(1)} = \ttheta^{(0)} + \tfrac{1}{\alpha} \sum_{j \in \mathcal{V}} \tilde{\ttheta}_j$. Next, define the $p$-dimensional difference vectors $\{ \ddelta_{j} \}_{j \in \mathcal{V}}$ with elements:
\begin{eqnarray}
(\ddelta_{j})_{j'} & = &
\left\{
\begin{array}{rcr}
(\ttheta_{j'})_{j} - (\ttheta_{j})_{j'} & \mbox{ if } & j \neq j'
\\
-(\ttheta_{j})_{j'} &  \mbox{ if } & j=j',
\end{array}
\right.
\end{eqnarray}
for all $j,j' \in \mathcal{V}$. Let $L( \cdot; \ttheta^{(0)})$ be the second-order Taylor approximation of $\mathcal{L}_{\mbox{{\tiny penPL}}}$ at $\ttheta^{(0)}$. Then $L(\ttheta^{(1)}; \ttheta^{(0)}) > L(\ttheta^{(0)}; \ttheta^{(0)})$ if,
\begin{eqnarray} \label{supp_max_thm_non_conv_condition}
\alpha & \geq & \alpha_{min} = 3 + \frac{3}{2} \cdot \frac{\sum_{j \in \mathcal{V}} \ddelta_j^\top \mathbf{H}_{j} \ddelta_j}{\sum_{j \in \mathcal{V}} \ttheta_j^\top \mathbf{H}_{j} \ttheta_j},
\end{eqnarray}
where $\mathbf{H}_j$ is the $j$-th block Hessian matrix.
\end{lemma}

\begin{proof}
Refer to Supplementary Material C.
\end{proof}

\noindent
Lemma \ref{supp_thm_non_conv_comb} presents a lower bound $\alpha_{min} > 3$ on $\alpha$ which warrants, when $\alpha > \alpha_{min}$, an increase of the penalized pseudo-likelihood $\mathcal{L}_{\mbox{{\tiny penPL}}}(\mathbf{\Theta},\mathbf{Y}_1,\ldots,\mathbf{Y}_n)$ at each iteration of Algorithm \ref{alg.pseudocode}. However, this lower bound is computationally demanding to evaluate as it depends on $\ttheta_j$ and $\mathbf{H}_j$. Therefore, this result mainly motivates the use of an $\alpha$ smaller than $p$. Then, whenever the algorithm does not decrease the error $\| \left. \partial \mathcal{L}_{\mbox{{\tiny penPL}}}/ \partial \ttheta \right|_{\ttheta=\hat{\ttheta}^{(k)}} \|_2$, the employed $\alpha$ is reset to (say) twice its value. Eventually, the multiplier $\alpha$ then becomes sufficiently large to ensure an improvement of the loss function and, thus, guarantees the convergence of Algorithm \ref{alg.pseudocode}. In practice, we implemented and used $\alpha_{min}$ throughout as it significantly speeds up the convergence of Algorithm \ref{alg.pseudocode}.

% \newpage

\subsection{Implementation}
\subsubsection{Algorithm}
Algorithm \ref{alg.pseudocode} was implemented in C++ using the OpenMP API that supports multithreading with a shared memory. For convenience of the user, the algorithm was wrapped in an R-package as extension for the R statistical computing software, together with some extensions such as $k$-fold cross-validation and a Gibbs sampler to draw samples from the pairwise MRF distribution. The package will be made publicly available on GitHub. The boundary constraints on the parameter space (e.g. those that ensure well-definedness for the GLM family) are implemented using an additional convex and twice-differentiable penalty whenever one of the boundary constraints is violated.

\subsubsection{Cross-validation}
The penalty parameter $\lambda$ of the ridge pseudo-likelihood estimator is selected using $k$-fold cross-validation. This amounts to dividing the $n$ samples over $k$ exhaustive and mutually exclusive groups. The samples from all-but-one of these groups are used to compute the estimator for a given choice of the penalty parameter. The performance of the obtained estimator is evaluated on the samples of the left-out group. This process is repeated $k$ times, leaving each group out once, but using the same penalty parameter value throughout. The resulting $k$ performances are averaged and this average is the estimated performance of the estimator for the employed penalty parameter value. The performance is assessed for a grid of penalty parameters usually spanning multiple orders of magnitude (e.g. $\lambda\in[10^{-10},10^2]$). The value that yields the best estimated performance is considered optimal and used to obtain the final, optimal estimator. In the remainder the performance is evaluated with the mean squared prediction error (MSPE) and $k=10$ unless stated otherwise. The choice for $k=10$ follows the recommendation provided in \cite{Friedman2001}, where it is suggested to be reasonable compromise between the resulting bias and variance of the cross-validated prediction error.

The simpliciy and generality of cross-valdation makes it an obvious choice for penalty parameter selection. Moreover, as cross-validation chooses a model with good prediction performance, it is also a natural choice for ridge-type estimators that do not induce sparsity. Alternatively, one may employ an information criterion, that balances model fit and  complexity, for penalty parameter selection. The extended Bayesian information criterion (eBIC) has been proposed for model selection with high-dimensional data \citep{chen2008extended}. \cite{gao2010composite} generalized the eBIC for the sparse maximum pseudo-likelihood estimators. Here the application of eBIC for penalty parameter selection is currently hampered by the lack of a proper definition of model complexity for non-sparse ridge-type estimators. But when available, it may be a promising computationally less demanding alternative to cross-validation.

\subsubsection{Sparsification}
The maximum ridge pseudo-likelihood estimate of $\mathbf{\Theta}$ does not contain any zero's. Zero's can be inferred by a post-estimation sparsification procedure. Hereto we propose to use the off-the-shelf empirical Bayes methodology of \cite{efron2004}. It assumes that absolute value of the off-diagonal elements of the standardized version of $\widehat{\mathbf{\Theta}}$ are i.i.d. following a two-component mixture: $f(\cdot) = \pi_0 f_0(\cdot) + (1-\pi_0) f_1(\cdot)$ with mixing proportion $\pi_0$ and components $f_0(\cdot)$ and $f_1(\cdot)$. The first component, $f_0(\cdot)$, represents the distribution of the off-diagonal elements of the standardized $\widehat{\mathbf{\Theta}}$ corresponding to edges absent in the underlying conditional independence graph, while the second component, $f_1(\cdot)$, relates to its present edges. The parameter $\pi_0$, and densities $f(\cdot)$ and $f_0(\cdot)$ are all estimated, employing some convenient assumptions, from the histogram of the off-diagonal elements of the standardized version of $\widehat{\mathbf{\Theta}}$. With these estimates at hand the posterior probability of an edge in the underlying graph given the observed element of the standarized $\widehat{\mathbf{\Theta}}$, i.e. $P((j, j') \in \mathcal{E} \, | \, (\tilde{\mathbf{\Theta}}_{j,j'})$, is calculated. This probability can be endowed with a local false discovery rate interpretation (see \citealp{efron2004}). This probability is used as the basis for sparsification. If large, an edge is inferred. If small, none is concluded.

\subsubsection{Sampling}\label{sect.sampling}
To draw data from the full pairwise MRF distribution may be difficult. But as the conditional distributions of each variate given the others are known explicitly and of relatively simple functional form, a Gibbs sampler to draw samples from the pairwise MRF distribution is easily constructed. It requires sampling from univariate exponential family members. Details of the resulting Gibbs sampler are immediate from its pseudo-code given in Supplementary Material D. In the remainder this algorithm is applied with a burn-in period of length $5.000$ and a thinning factor $1/500$, such that samples are selected from the chain $500$ iterations apart ensuring independence among samples (\citealp{Chen2015}).

\section{Simulations}
We evaluate the performance of the parallel block coordinate Newton-Raphson Algorithm \ref{alg.pseudocode} for the numerical evaluation of the maximum penalized pseudo-likelihood estimator together with the quality of the resulting estimator in a numerical study with synthetic data. Throughout we use the convex and twice differentiable ridge penalty $\|\mathbf{\Theta}\|_F^2$. Hence, it is the quality of the ridge pseudo-likelihood estimator of $\mathbf{\Theta}$ that is studied. The diagonal of the parameter matrix $\mathbf{\Theta}$ is left unpenalized (as recommended by \citealp{Tibshirani2009}). Throughout this section, the cross-validated ridge pseudo-likelihood estimator $\widehat{\mathbf{\Theta}}_n(\lambda_{\mbox{{\tiny opt}}})$ of $\mathbf{\Theta}$ is learned with Algorithm \ref{alg.pseudocode} using threshold $\tau = 10^{-10}$ and multiplier $\alpha= \alpha_{min}$ unless stated otherwise.

\subsection{Performance illustration}
We illustrate the performance of the algorithm and the estimator within a simulation assuming a lattice graph $\mathcal{G} = (\mathcal{V},\mathcal{E})$, thus following \citealp{yang2014mixed}, \citealp{LeeHastie2013}, and \citealp{Chen2015}. The chosen graph's layout represents the most general setting encompassed by the outlined theory in which each GLM family member is present with an equal number of variates (Figure \ref{fig.lattice_graph}a). The interactions obey the parameter restrictions for well-definedness of the pairwise MRF distribution. The resulting lattice graph for $p=16$ nodes has $|\mathcal{E}|=36$ edges. The corresponding pairwise MRF distribution has $136$ unique parameters and allows for $120$ edges. Hence, the constructed network is dense in the sense that it contains $30\%$ of all possible edges. Consequently, the nodes have an average degree of $4.5$, while correct graphical model selection is no longer guaranteed (asymptotically) when the maximum vertex degree is larger than $\sqrt{p/\log (p)} = \sqrt{16/\log (16)} = 2.4$ (\citealp{das2012learning}). The employed lattice graph thus represents a setting where previous work on (sparse) graphical models with data of mixed types fails when the sample size is small, relative to the number of parameters. To ensure the resulting pairwise MRF distribution $P_{\mathbf{\Theta}}(\mathbf{Y})$ adheres to the described lattice graph $\mathcal{G}$, we choose its parameter $\mathbf{\Theta}$ as follows:
\begin{eqnarray*}
\mathbf{\Theta}_{j, j'} & = & \left\{ \begin{array}{cl}
-0.2 & j, j' \in\mathcal{V} \mbox{ such that } j \neq j' \mbox{ and }  (j,j')\in\mathcal{E},
\\
-0.2 & j, j' \in\mathcal{V}  \mbox{ such that } j = j' \mbox{ and } Y_{j}  \mbox{ follows}
\\ 
& \mbox{either a Bernoulli or an exponential},
\\
2 & j, j' \in\mathcal{V}  \mbox{ such that } j = j' \mbox{ and } Y_{j} \mbox{ follows}
\\ 
& \mbox{either a Gaussian or a Poisson},
\\
0 & \mbox{Otherwise.}
\end{array}
\right.
\end{eqnarray*}
This parameter choice ensures the pairwise MRF distribution $P_{\mathbf{\Theta}}(\mathbf{Y})$ is well-defined and all edges share the same edge weight (as captured by the nonzero off-diagonal elements of $\mathbf{\Theta}$). Moreover, the variance of all node-conditional Gaussian variables was fixed at $\sigma^2=1$. With the distribution fully specified data are generated with the Gibbs sampler (described in Supplementary Material D) using the node-conditional distributions derived from $P_{\mathbf{\Theta}}(\mathbf{Y})$.

% \newpage 
\begin{figure*}[h!]
\begin{center}
\centering
\begin{tabular}{c}
{\includegraphics[angle=0, scale=0.8,clip,trim=35 415 50 35]{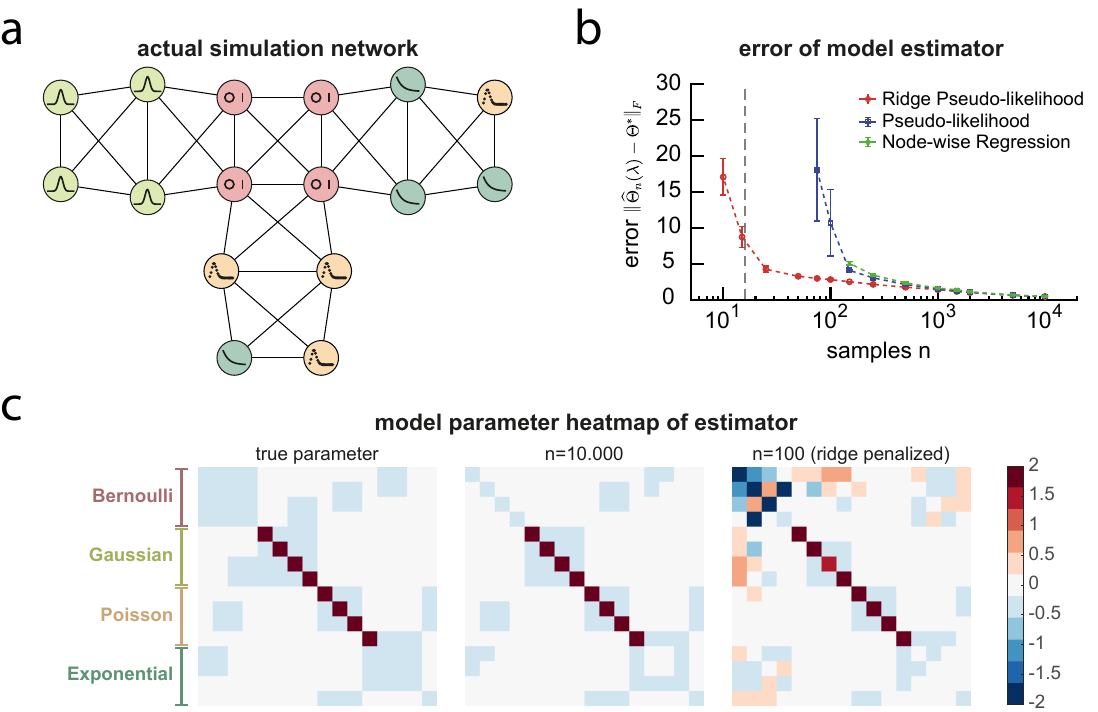}}
\end{tabular}
\end{center}
\vspace{-0.55cm}
\caption{\textbf{a,} The synthetic simulation lattice graph $\mathcal{G}=(\mathcal{V},\mathcal{E})$. Nodes with variates distributed in accordance with the GLM family members Bernoulli, Gaussian, Poisson and exponential are represented by a pictogram representing (the shape of) their distribution. \textbf{b,} Scaling of the estimator errors with the sample size. Included are the cross-validated ridge pseudo-likelihood (red), its unpenalized counterpart (blue) and the `averaged' node-wise regression coefficients (green). The number of replicates is inversely proportional to the sample size ($3$ replicates for $n=10^4$, and subsequently $6, 15, 18, 30, 60, 120, 180, 200, \ldots$ replicates). Error bars denote a $95\%$ confidence interval for the mean. \textbf{c,} Heatmaps of the true parameter (left) and the ridge pseudo-likelihood estimators of the lattice graph for $n=10.000$ (middle) and $n=100$ samples (right). 
}
\label{fig.lattice_graph}
\end{figure*}

% \newpage 
\begin{figure*}[h!]
\begin{center}
\centering
\begin{tabular}{c}
{\includegraphics[angle=0, scale=0.8,clip,trim=40 360 40 40]{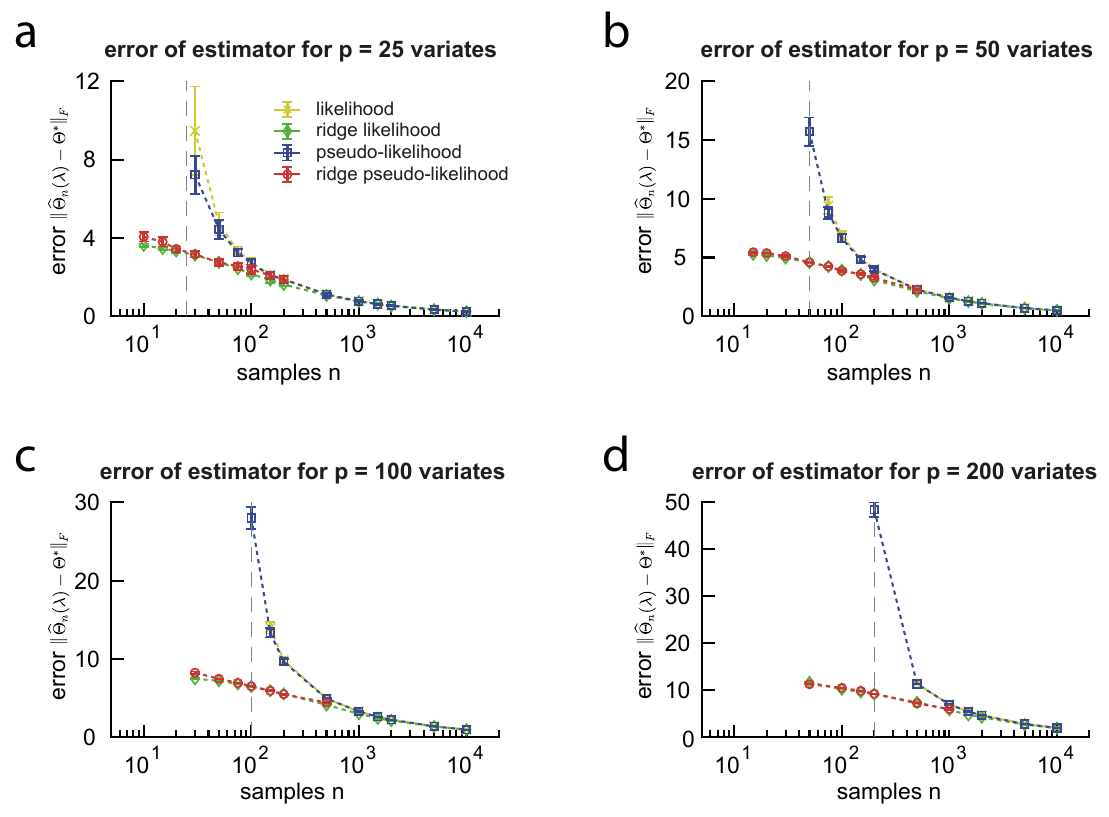}}
\end{tabular}
\end{center}
\vspace{-0.55cm}
\caption{Scaling of the precision matrix estimator error with the sample size. Included are the ridge pseudo-likelihood (red), its unpenalized counterpart (blue), the ridge likelihood (green) and its unpenalized counterpart (yellow). The number of replicates are $10$ for $n=10^4$, and subsequently $20$ for $n < 10^4$. Error bars denote the standard deviation. The number of variates of the model are \textbf{a,} $p=25$, \textbf{b,} $p=50$, \textbf{c,} $p=100$ and \textbf{d,} $p=200$.
}
\label{fig.simuResults}
\end{figure*}

% \newpage
% \noindent
From the generated data drawn from the `lattice graph' distribution we calculate the cross-validated ridge pseudo-likelihood estimator $\widehat{\mathbf{\Theta}}_n(\lambda_{\mbox{{\tiny opt}}})$ of $\mathbf{\Theta}$ using Algorithm \ref{alg.pseudocode}. Next, we evaluate the performance of Algorithm \ref{alg.pseudocode} on data for $n \in [10, 10^4]$. To this end, we compare the error of the cross-validated ridge pseudo-likelihood estimator to its unpenalized counterpart and the averaged node-wise regression coefficients as baseline (whenever the sample size allowed for the evaluation of the latter two). Here the error is defined as the Frobenius norm of the difference between the parameter and its estimate  $\| \widehat{\mathbf{\Theta}}(\lambda_{\mbox{{\tiny opt}}})-\mathbf{\Theta}\|_F$, and analogous for the other two. The error of each estimator is determined as function of the sample size (Figure \ref{fig.lattice_graph}b). The error of the ridge pseudo-likelihood estimator $\widehat{\mathbf{\Theta}}(\lambda_{\mbox{{\tiny opt}}})$ decreases slowly with the sample size $n$ in the low-dimensional regime as expected, while a sharp increase of its error of is observed in a high-dimensional setting. In the low-dimensional regime the error of the ridge pseudo-likelihood is generally on a par with its unpenalized counterpart and the baseline node-wise regression. More refined, both the maximum ridge and unpenalized pseudo-likelihood estimator outperform the baseline averaged node-wise regression for all sample sizes. A full information and simultaneous parameter estimation approaches are thus preferable. Finally, the proposed ridge pseudo-likelihood estimator clearly shows better performance in the sample domain of (say) $n<150$. Hence, regularization aids (in the sense of error minimization) when the dimension $p$ approaches or exceeds the sample size $n$. Finally, representative examples of resulting estimates are visualized by means of a heatmap with the true parameter as reference (Figure \ref{fig.lattice_graph}c). Strikingly, the ridge pseudo-likelihood estimator has most deviations from the true parameter for 'Bernoulli nodes', predominantly amongst themselves. This hints at the need for a larger sample size in order to adequately estimate the parameters related to interactions among Bernoulli variates.

\subsection{Comparison}
In an effort to compare the performance of the proposed ridge pseudo-likelihood estimator, we consider the Gaussian graphical model and compare with the ridge precision estimator (\citealp{van2016ridge}). Assuming all variates being jointly normal, $\mathbf{Y} \sim \mathcal{N}(\mathbf{0}_p, \mathbf{\Omega}^{-1})$, the latter estimates $\mathbf{\Omega}$ through ridge penalized likelihood maximization. The ridge pseudo-likelihood estimator too estimates $\mathbf{\Omega}$, but does so in a limited information approach. Here we compare the quality of these full and limited information approaches \textit{in silico}. To this end define a three-banded precision matrix $\mathbf{\Omega}$ with a unit diagonal, $(\mathbf{\Omega}_{j,j+1} = 0.5 = (\mathbf{\Omega})_{j+1,j}$ for $j=1, \ldots, p-1$, $(\mathbf{\Omega}_{j,j+2} = 0.2 = (\mathbf{\Omega})_{j+2,j}$ for $j=1, \ldots, p-2$, $(\mathbf{\Omega}_{j,j+3} = 0.1 = (\mathbf{\Omega})_{j+3,j}$ for $j=1, \ldots, p-4$, and all other entries equal to zero. Here we let the number of variates $p$ range from $p=25$ to $p=200$ to test the performance of the proposed estimator for its intended use in the context of a large number of variates. We draw a sample of various sizes $n$ from the thus defined multivariate normal $\mathcal{N}(\mathbf{0}_p, \mathbf{\Omega}^{-1})$. From these data both the likelihood and pseudo-likelihood estimators are evaluated.

\noindent
We compare the performance of the two ridge precision estimators by means of the error, defined as $\lvert\lvert\widehat{\mathbf{\Omega}}(\lambda_{\mbox{\tiny{opt}}})-\mathbf{\Omega}\rvert\rvert_{F}$, to their unpenalized analogues (Figure \ref{fig.simuResults}).  First we consider a low number of variates $p=25$ (Figure \ref{fig.simuResults}a). In the low-dimensional regime the errors of all estimators are very close and decrease slowly with the sample size $n$ as expected. In the high-dimensional regime ($n < 100$), expectedly, the penalized estimators clearly outperform their unpenalized counterparts as can be witnessed from their diverging error when $n$ approaches $p$. Restricted to the penalized estimators, the ridge likelihood estimator appears to outperform its ridge pseudo-likelihood counterpart slightly. This is probably to the full information usage of the likelihood. More importantly, the performance (as measured by the error of the precision matrix) of the maximum penalized pseudo-likelihood estimator is comparable to that of the maximum penalized likelihood estimator. This corroborates the results of previous simulation studies into the maximum (lasso) penalized pseudo-likelihood estimator  (\citealp{LeeHastie2013}, \citealp{Tibshirani2009}). (The error of the estimators in the high-dimensional regime is further studied in Supplementary Material E, figure 1). With the application of large data sets and parallel computing in mind, we consider the performance of the maximum penalized pseudo-likelihood estimator for a higher number of variates up to $p=200$ next (Figure \ref{fig.simuResults}b-d). Generally, while an increase of the dimension $p$ increases the error of the estimators, qualitatively their relative behavior is largely unchanged. More specifically, for $p=50$, $p=100$ and $p=200$ the errors of all estimators are very close in the low-dimensional regime, similarly to the `$p=25$'-case. In the high-dimensional regime the errors of the maximum penalized pseudo-likelihood estimator are again comparable to that of the maximum penalized likelihood estimator.

\begin{figure*}[b!]
\begin{center}
\centering
\begin{tabular}{c}
{\includegraphics[angle=0, scale=0.8,clip,trim=50 360 50 50]{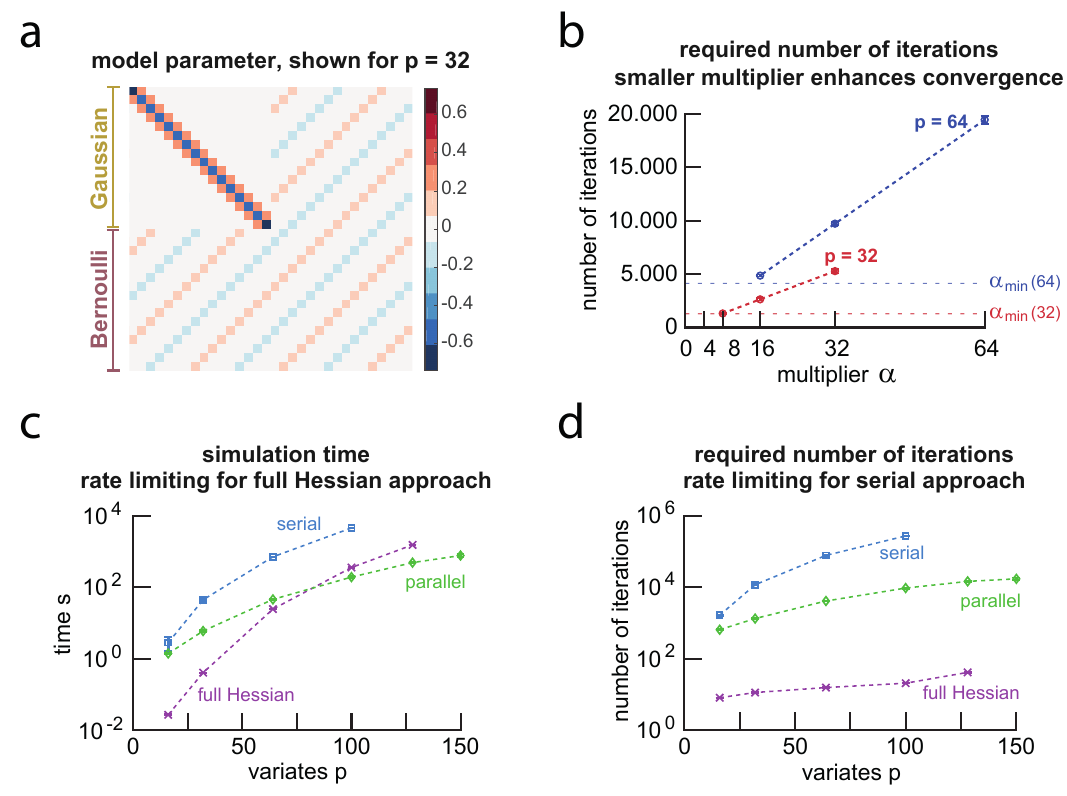}}
\end{tabular}
\end{center}
\vspace{-0.5cm}
\caption{Details and results of the benchmark study of the runtime and required number of iterations for convergence of Algorithm \ref{alg.pseudocode}. \textbf{a,} Heatmap of the Bernoulli-Gaussian MRF distribution parameter (for $p=32$, other number of variates are defined analogous). \textbf{b,} The required number of iterations plotted against the multiplier $\alpha$ for $p=32$ (red) and $p=64$ (blue) variates. Lightblue and lightred dotted lines denote the number of iterations required when using the multiplier $\alpha_{min}$. \textbf{c,} and \textbf{d,} The simulation time (\textbf{c}) and required number of iterations (\textbf{d}) for Algorithm \ref{alg.pseudocode} to converge with treshold $\tau=10^{-10}$ plotted against the number of variates ($p=16, 32, 64, 100, 128$ and $150$). Included are the naive Newton-Raphson approach (purple) and the serial (blue) and parallel (green) block coordinate approaches. All simulations were run on a standard laptop used for scientific work using four processing cores for the parallel algorithm. The inverse Hessian matrices was computed every $k=p$ iterations to reduce computation time for each method. Each condition has $5$-replicates and error bars denote the standard deviation. 
}
\label{fig.benchmark}
\end{figure*}

\subsection{Speed-up and benchmark}
Here we pursue to speed up Algorithm \ref{alg.pseudocode} by further reducing its computational complexity. First, the block-coordinate Newton-Raphson is modified to a block-coordinate quasi-Newton approach by using a chord method that computes inverse of the block-wise Hessian matrices $\big\{\mathbf{H}_{j}\big\}_{j\leq p}$ only every $k_0=p$ iterations of the algorithm. This vastly reduces the computational complexity of the algorithm by alleviating the effect of the rate-limiting step.

We benchmark the proposed algorithm by studying its runtime and the required number of iterations. This is done for a distribution with $p$ variates of the binary (Bernoulli) and continuous (Gaussian) type. The two types are equally represented among the $p$ variates. The parameter $\mathbf{\Theta}$ of this distribution is chosen such that it satisfies the parameter restrictions for well-definedness of the pairwise MRF distribution.  The conditional precision matrix of the Gaussian variates is three-banded as in the previous comparison study. Each Bernoulli variate has an interaction with every fifth other variate, and the corresponding interaction parameters are set equal to $\pm 0.1$, in alternating fashion. The resulting MRF distribution parameter $\mathbf{\Theta}$ has about $20\%$ non-zero elements and is visualized by a heatmap in Figure \ref{fig.benchmark}a. Data are drawn from the mixed Bernoulli-Gaussian distribution with the Gibbs sampler. Each dataset contains $n=1.000$ samples with a dimension ranging from $p=16$ to $p=150$.

The effect of the choice of the multiplier $\alpha$ on the required number of iterations for convergence in the evaluation of the pseudo-likelihood estimator $\widehat{\mathbf{\Theta}}_n$ on the data from the Bernoulli-Gaussian distribution with $p=32$ and $p=64$ variates is 
shown in Figure \ref{fig.benchmark}b. A multiplier $\alpha = p$ corresponds to the naive approach of averaging block coordinate updates (a convex combination). Reducing $\alpha$ -- and thereby increasing the step-size -- improves the rate of convergence of Algorithm \ref{alg.pseudocode}, especially for large $p$. In fact, for fixed $p$ the convergence rate of Algorithm \ref{alg.pseudocode} scales nearly linearly with the choice for $\alpha$. However, the algorithm diverges for too small $\alpha$. This is prevented by choosing $\alpha \geq \alpha_{min}$ with the proposed $\alpha_{min}$ from Lemma \ref{supp_thm_non_conv_comb} to ensure convergence and the optimization of the required number of iterations (Figure \ref{fig.benchmark}).

The runtime and required number of iterations for termination of Algorithm \ref{alg.pseudocode} with the implemented multiplier $\alpha_{min}$ are presented in Figure \ref{fig.benchmark}c-d. The figure contains the results for the pseudo-likelihood estimator $\widehat{\mathbf{\Theta}}_n$ computed in three ways: \textit{i)} using the naive Newton-Raphson algorithm, the \textit{ii)} serial and \textit{iii)} parallel block coordinate approach. The naive Newton-Raphson approach computes and inverts the full Hessian matrix. The serial block coordinate approach updates $\widehat{\mathbf{\Theta}}_n$ for one block $j<p$ per iteration and then inverts only the $j$-th Hessian $\mathbf{H}_{j}$ matrix. The parallel approach is Algorithm \ref{alg.pseudocode} and computes and inverts all $\big\{\mathbf{H}_{j}\big\}_{j\leq p}$ in parallel. For a fair comparison each approach inverts their respective Hessian matrices only every $k_0=p$-th step. Note that using the naive Newton-Raphson that inverts the full Hessian at each step is computationally too intensive, while a diagonal Hessian requires too many iterations for convergence (for a comparison see Supplementary Material E, Figure 2). The results reveal that for small $p$ the naive Newton-Raphson approach outperforms the block coordinate approaches. This is due the quadratic convergence and the computational tractability of the inversion of the Hessian matrix. For large $p\geq100$ the parallel block coordinate approach is fastest as its runtime increases slowest with $p$. In particualar, the required number of iterations to termination of the serial block coordinate approach increases very fast with $p$, while the computational complexity of inversion the full Hessian matrix quickly becomes prohibitly large. To appreciate the computational efficiency of Algorithm \ref{alg.pseudocode} further note that one iteration of the naive Newton-Raphson implementation has computational complexity $O(p^6)$ compared to $O(p^3)$ of the parallel Algorithm \ref{alg.pseudocode}. This permits the latter $O(p^3)$ iterations before exceeding the computational complexity of one iteration of the naive implementation. Hence, the full Hessian approach will always be outperformed by its alternatives due to the inversion of the full Hessian matrix, irrespective of its implementation. The presented algorithm was found to always terminate within $O(p^3)$ iterations.

In summary, the parallel block coordinate approach balances the required number of iterations and the time required per iteration for all approaches, thereby optimizing the total runtime. Our proposed algorithm can thus be considered as exploiting the best of both worlds: using full second-order information, while minimizing computational complexity.

% \clearpage
% \newpage
\section{Application} \label{sect.application}
The proposed methodology is illustrated by means of learning the pairwise MRF distribution of (part of the) cellular regulatory network from omics data. The illustration uses the publicly available non-silent somatic mutation and gene expression data from the invasive breast carcinoma study of the Cancer Genome Atlas (TCGA) project (\citealp{TCGA2012}). The gene expression data comprises the RNA sequencing profiles of $445$ patients with invasive breast carcinomas (BRCA) and are obtained via the {\tt brcadat} dataset included in the R-package {\tt XMRF} (\citealp{XMRF2016}). The thus obtained data set includes the expression levels (normalized mRNA read counts) of $353$ genes listed in the Catalogue of Somatic Mutations in Cancer (COSMIC) (\citealp{Sanger2017}). Subsequently, the RNA sequencing data are preprocessed as described in \citealp{allen2013local} using the {\tt processSeq} function from the {\tt XMRF} package. This results in expression levels that -- presumably -- can be modeled with a Poisson distribution. Then, the top $15\%$ genes with largest variance in expression levels across patients are retained. Next, the mutation data of $977$ BRCA patients are obtained from the TCGA project and filtered with the {\tt getTCGA}-function from the R-package {\tt TCGA2STAT} (\citealp{TCGA2STAT2016}). The binary data indicate the presence of a mutation in the coding region of a gene. Furthermore, only mutation data of genes with mutations present in at least $5\%$ of the patients are included. Finally, the mutation and expression data sets are merged at the gene level and for common patients using the {\tt OMICSBind} function from the {\tt TCGA2STAT}-package. The resulting final data matrix contains $p=63$ variates comprising the non-silent somatic mutations of $11$ genes (binary) and the expression level of $52$ genes (counts) from $n=433$ BRCA patients. The genes \emph{CDH$\it 1$} and \emph{GATA$\it 3$} have both mutation and expression data measured.

The pairwise MRF distribution with node-conditional Bernoulli and Poisson distributions for mutation and expression variates, respectively, was fitted to the data described above using Algorithm \ref{alg.pseudocode} to find the proposed ridge pseudo-likelihood estimator. The penalty parameter is selected with $5$-fold cross-validation using a grid search over a range of penalty parameters to minimize the mean square prediction error.

\begin{figure*}[b!]
\begin{center}
\centering
\begin{tabular}{c}
{\includegraphics[angle=0, scale=0.75,clip,trim=30 450 30 30]{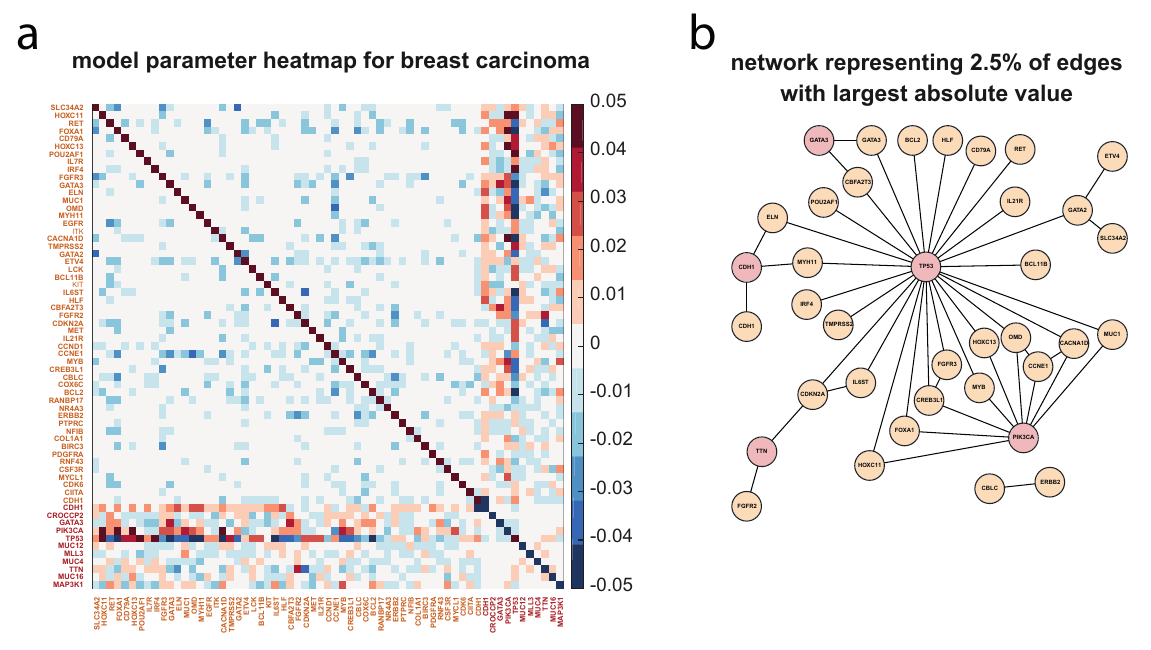}}
\end{tabular}
\end{center}
\vspace{-0.5cm}
\caption{\textbf{a,} Heatmap of the ridge pseudo-likelihood estimator of the invasive breast carcinoma gene network. Non-silent somatic mutations in genes are labeled red, gene expression levels are labeled yellow. Positive interactions between genes are shaded red, negative interactions are shaded blue. Algorithm (\ref{alg.pseudocode}) was run using threshold $\tau = 10^{-3}$ and  terminated within $34.000$ iterations for each run. \textbf{b,} Network visualization of the ridge pseudo-likelihood estimator of the BRCA gene network. Shown are the connected components of the network with edges whose absolute edge weight is among the largest $2.5\%$ in absolute sense (disconnected nodes are omitted). Nodes in the network are labeled with the gene name, where either non-silent somatic mutations (binary, red) or gene expression levels (counts, yellow) were measured.
}
\label{fig.brca_heatmap}
\end{figure*}

The resulting estimate is visualized by means of a heatmap (Figure \ref{fig.brca_heatmap}a). Interaction parameters between genes' expression levels are always negative, in line with the constraints for well-definedness. Moreover, the majority of strongest interactions (in absolute sense) relate gene expression levels to somatic mutations. Finally, we summarized the ridge pseudo-likelihood estimator as a network for interpretation purposes (Figure \ref{fig.brca_heatmap}b). The network includes only the $2.5\%$ edges with the largest (in the absolute sense) weights (this simple selection procedure only facilitates the network illustration and is not aimed at any form of type I error control). These plotted edges include two types of interactions, interactions among genes' expression levels and those between one gene's expression levels and the somatic mutations of other genes. The majority of these edges involve the non-silent somatic mutation in \emph{TP$\it 53$}, which is a well-known tumor suppressor gene (\citealp{Levine1991}). These edges relate a mutation in \emph{TP$\it 53$} to a change in the expression level of genes that have been causally implicated with cancer (\citealp{Sanger2017}), thereby confirming the role of \emph{TP$\it 53$} as guardian of the genome (\citealp{Lane1992}). Another biomarker with a large node degree is a non-silent mutation in \emph{PIK$\it 3$CA}, a known oncogene that is often mutated in breast cancers (\citealp{Cizkova2012}). Otherwise noteworthy are genes \emph{CDH$\it 1$} and \emph{GATA$\it 3$}, the only ones included with both mutation and expression data. The molecular levels of these two genes are connected by an edge in the displayed network, which indicates that their expression levels are related to a mutation in their DNA template. Finally, the network contains no edges between the binary mutation variates, which may have been expected as somatic mutations may co-occur but are generally believed to be unrelated.

\section{Conclusion}
We presented methodology for the estimation of multivariate exponential family distributions. As special case of interest, the employed class of distributions encompasses the pairwise Markov random field that describes stochastic relations among variates of various types.

The model parameters are estimated by means of penalized pseudo-likelihood maximization to account for collinearity and an excess of variates (relative to the sample size). This estimator was shown to be consistent under mild conditions. Our algorithm allows for efficienct computation on multi-core systems and accommodates for a large number of variates. The algorithm was shown to converge and terminate. The performance of the penalized estimator was compared to related ones, and the proposed algorithm was benchmarked in a simulation study. The latter resulted in several improvements to enhance computational efficiency and speed of convergence. Finally, our methodology was demonstrated with an application to an integrative omics study using data from various molecular levels (and types), which yielded biological sensible results.

Envisioned extensions of the presented ridge pseudo-likelihood estimator allow -- among others -- for variate type-wise penalization. Technically, this is a minor modification of the algorithm but brings about the demand for an efficient penalty parameter selection procedure. Furthermore, when quantitative prior information of the parameter is available it may be of interest to accommodate shrinkage to nonzero values.

Foreseeing a world with highly parallelized workloads, our algorithm provides a first step towards a theoretical framework that allows for efficient parallel evaluation of (high-dimensional) estimators. Usually and rightfully most effort concentrates on the mathematical optimization of the computational aspects of an algorithm. Once that has reached its limits, parallelization may push further. This amounts to simultaneous estimation of parts of the parameter followed by careful -- to ensure convergence -- recombination to construct a fully updated parameter estimate. Such parallel algorithms may bring about a considerable computational gain. For example, in the presented case this gain was exploited to incorporate full second-order information without inferior computational complexity compared to existing algorithms.

\bibliographystyle{plainnat}
% \bibliography{/media/wessel/FREECOM HDD1/Wessel/Research/Articles/ridgeMRF/SAC/TEX/ridgeMRF.bib}

\bibliography{ridgeMRF.bib}

\end{document}